\documentclass{article}
\usepackage[utf8]{inputenc}

\def\notes{1}
\usepackage{style}

\newcommand{\cD}{\mathcal{D}}
\renewcommand{\H}{\mathcal{H}}
\newcommand{\Oh}{\mathcal{O}}
\newcommand{\Alg}{\mathrm{Alg}}
\newcommand{\Test}{\mathrm{Test}}

\title{The Generic Holdout:\\Preventing False-Discoveries in
Adaptive Data Science}
\author{
Preetum Nakkiran\thanks{Harvard John A. Paulson School of Engineering and Applied Sciences, Harvard University, 33 Oxford Street,
Cambridge, MA 02138, USA. Email: {\tt preetum@cs.harvard.edu}. Work supported in part by a Simons Investigator Award, NSF Award CCF 1715187, and the NSF Graduate Research Fellowship Grant No. DGE1144152.}
\and
Jaros\l aw B\l asiok \thanks{Harvard John A. Paulson School of Engineering and Applied Sciences, Harvard University, 33 Oxford Street,
Cambridge, MA 02138, USA. Email: {\tt jblasiok@g.harvard.edu.} Supported by ONR grant N00014-15-1-2388.}
}
\date{}

\begin{document}

\maketitle

\begin{abstract}

The traditional framework of science is ``non-adaptive'', in the sense that the
scientist first fixes a hypothesis, and then collects data to test it.
However, modern science is often ``adaptive'', in that first large amounts of data are collected, then the scientist
explores this data to propose hypotheses.
Such adaptive data analysis has posed a challenge to science due to its ability to generate false hypotheses on moderately large data sets. In general, with non-adaptive data analyses (where the  queries to the data are generated without being influenced by answers to previous queries) a data set containing $n$ samples may support exponentially many queries in $n$.
This number reduces to linearly many under naive adaptive data analysis,
and even sophisticated remedies such as the Reusable Holdout (Dwork et. al 2015) only allow quadratically many queries in $n$.

In this work, we propose a new framework for adaptive science which
exponentially improves on this number of queries under a restricted yet
scientifically relevant setting in data analysis, where the goal of the scientist is to find a single (or a few) true hypotheses about the universe based on the samples.
Such a setting may describe the search for predictive factors of some disease
based on medical data, where the analyst may wish to try a number of predictive models until a satisfactory one is found.

Our solution, which we refer to as the \emph{Generic Holdout} methodology, involves two extremely simple ingredients: (1) a partitioning of the data into a exploration set and a holdout set (a methodology that is already widely practiced) and (2) a limited exposure strategy for the holdout set (something that is widely violated, but easy to fix).
An analyst is free to use the exploration set arbitrarily, but when testing hypotheses against the holdout set, the analyst only learns the answer to the question: ``Is the given hypothesis true (empirically) on the holdout set?''
It is common for holdout sets to provide much more information, such as ``how well'' the hypothesis fit the holdout set, and this is where a common fallacy lies.
The resulting scheme is immediate to analyze, and reverts to the setting where exponentially many hypotheses can be tested against the data. Despite this simplicity we do not believe our method is obvious, as evidenced by the many violations in practice.

Our proposal is best seen as an alternative to \emph{pre-registration},
where journals require scientists to commit to their hypotheses and analysis procedure before collecting any data.
Pre-registration preserves statistical validity, but at the cost of adaptivity: researchers are not allowed to explore the structure of data to generate hypotheses.
In this setting, the Generic Holdout allows researchers to get the benefits of adaptive data analysis, without the problems of adaptivity.

\end{abstract}

\section{Introduction}
In science, it is natural to first collect data, and then form informed hypotheses based on it. This is arguably how much of science was historically done --- after all, one is unlikely to come up with a correct theory of physics without first observing the world.
In order for the result to be statistically valid, a scientist must then collect independent data, after fixing their hypothesis.
However, this is not done in practice: in modern experimental science,
it is common to first collect data, and then explore it to generate plausible hypotheses.
This is what we consider \emph{adaptive} science: wherein the scientist generates hypotheses after somehow interacting with the data set.
Doing this naively could lead to being convinced of false hypotheses, since a
scientist may ``overfit'' --- that is, derive a false hypothesis that appears to be true on
the data.
This could occur if the scientist implicitly or explicitly tests many hypotheses
before finding one that happens to fit the experimental data.
For example, if the first hypothesis tested proves false on the data set,
a scientist may want to use information about this first hypothesis test to revisit the data set,
and form a new hypothesis, and so on.

There is growing recognition in the sciences that this ``adaptive'' method of doing science
is not statistically sound, and leads to invalid claims.
There is a recognized ``reproducibility crisis'' in Psychology, for example, after a collaboration failed to replicate 62 out of 97 studies with positive findings
in three prominent psychology journals \cite{reproduce}.
The problem of adaptivity (also known as ``p-hacking'' or ``researcher degrees of freedom'')
is recognized as a key contributor to this crisis,
since it is often not correctly accounted for in standard statistical analyses \cite{crisis, flex-psych}.
Other experimental areas such as neuroscience~\cite{neuro-crisis}, economics~\cite{econ-crisis}, and social sciences~\cite{social-crisis} are also aware of this issue,
after conducting reproducibility studies following the example of Psychology.
In general, any methodology in which a researcher decides which hypothesis to test after somehow interacting with the data set is susceptible to this
problem of adaptivity.
Note that this includes scenarios like testing a second hypothesis after the initial hypothesis test failed, or doing a ``data-dependent analysis'' in which the hypothesis formed depends on some structure of the data set itself. One notable example of such data-dependent analysis is using Principal Component Analysis on the data set to find a correlation structure, and using this in turn to define the hypothesis, that is tested against the same data set.
The review article \cite{crisis} provides many further examples of this problem of adaptivity arising in science. We provide an explicit, formal example of this problem in Section~\ref{sec:problems}.

The preceding discussion motivates our scientific goal:
we would like to have a statistically-valid scientific methodology that allows researchers to explore their data before generating hypotheses.
{\bf We are considering the setting where the scientist is trying to derive a small number of true hypotheses, and would like to adaptively check if the proposed hypotheses are true.}
We want to guarantee that in this process, false hypotheses which are proposed are unlikely to be validated as true --- that is, we would like to prevent false discoveries.
Before describing our proposal, we first briefly discuss three proposed solutions to the problem of adaptivity in the sciences.

{\bf 1. Preregistration. } There has been a recent push in the scientific community towards preregistration: requiring scientists to commit to their scientific methods and hypotheses before conducting a study.
Several prominent journals now encourage scientists to preregister their research, in order to preserve statistical rigor (for example: \emph{Open Science} of the Royal Society,
and \emph{Psychological Science} of the Association for Psychological Science).
An open  letter published by more than 80 signatories calls for preregistration in the sciences ~\cite{open-letter}.
Preregistration does preserve statistical validity, but at the cost of adaptivity: the researcher is not allowed to explore the structure of data to generate hypotheses, and thus preregistration has been critiqued for slowing down the overall scientific process.

{\bf 2. Naive Holdout. }
Keeping a holdout set is another way of addressing the adaptivity problem.
The idea is, the scientist holds out part of the data set (without looking at it to form hypotheses), and is free to explore the remainder of the data set.
Then, after exploring and proposing a hypothesis, the scientist checks the hypothesis against the holdout set. This is statistically valid, since the holdout data is independent of the hypothesis being tested, but has the disadvantage that the holdout can only be used once:
if the scientist now wants to test another hypothesis, he/she must collect
additional independent data to use as a holdout.
This is sample-inefficient and impractical in settings where collecting data is expensive.
In particular, the naive holdout can only handle linearly many hypothesis tests
in the total size of the holdout sets.
Note that it is not valid to naively re-use the same holdout set
after seeing the results of the first hypothesis test (say, its $p$-value) --- this can quickly lead to overfitting on the holdout set,
in a way made precise in an extended example in Section~\ref{sec:problems}.

{\bf 3. Reusable Holdout. }
The Reusable Holdout ~\cite{reusable,reuse2}, a recent development in the field of Adaptive Data Analysis,
manages to improve on the naive holdout, and handles up to quadratically many hypothesis tests in the holdout size\footnote{See Section~\ref{sec:related} for a more complete discussion and comparison with the Reusable Holdout.}.
The insight of the Reusable Holdout was to leak less information than the naive holdout,
in part by only releasing a noisy estimate of how well the hypothesis fits the
holdout set (instead of, say, the exact $p$-value), thus preventing overfitting to the holdout set.

We extend the reusable holdout idea, and propose that the holdout set should in fact leak no information about the hypothesis test, except for what is absolutely necessary: whether the hypothesis passed or not on the holdout set. In particular, it should \emph{not} release any indication of ``how well'' the hypothesis fits the holdout set, such as a $p$-value.
By this simple modification, our proposal allows {\bf exponentially-many} adaptively-chosen false hypotheses to be invalidated, before the scientist discovers a true hypothesis.
Moreover, it comes at no cost: in our scientific setting, there is no reason to leak more information from the holdout set, since ultimately we are only interested in preventing false discoveries. If we want to report $p$-values for the confirmed hypothesis, this can be done by simply keeping another holdout set, to use after the Generic Holdout, specifically for the purpose of finding the $p$-values for validated hypotheses that will be published.

{\bf Proposed Method: The Generic Holdout. } To recap, our method is simple: first, the scientist collects the data, and keeps a holdout set for validation (without looking at it).
The scientist is free to explore the rest of the data (\emph{exploration set}) to come up with hypotheses
that he/she deems plausible. Each time the scientist proposes a hypothesis, the validation procedure only returns ``True'' or ``False'': whether the hypothesis passed validation on the holdout set, or not --- it should not release more information, like the $p$-value of the hypothesis. The scientist can revisit the rest of the data, and adaptively propose hypotheses, and continue until he/she proposes a hypothesis that is confirmed by validation.
In general, the scientist can continue this process until a small number of hypotheses are confirmed.

We stress that our method applies specifically to the case where the scientist's goal is to derive a small number of true hypotheses, and will stop adaptively proposing hypotheses once several of them are validated to be true.

These ideas are not technically novel, but we are not aware of the problem of how to prevent false discoveries in science being phrased and addressed with such minimal assumptions, and with the guarantees we provide.

\begin{remark}
Here we point out a benefit of the Generic Holdout, and clarify the sense in which it is ``adaptive,'' by contrasting it with a related method.

An alternative way of using the holdout data set to provide a statistically sound and sample efficient methodology is the following:
The data analyst interacts with the exploration set in an arbitrary way, and without looking at the holdout set produces a family of hypotheses, that is then simultaneously validated against the holdout set.
In fact, this scenario is technically very similar to the Generic Holdout (as described below), but practically very different.

Indeed, let us consider the following thought experiment. The scientist is using the exploration set to come up with a hypotheses, and yet whenever they have a hypothesis in mind, the scientist pretends that it has been invalidated on the holdout set, and proceeds to come up with the next hypothesis. Eventually, he/she would come up with a family of hypotheses independent of the holdout set, and could validate all of those simultaneously. This is exactly the same sequence that would be generated in the real interaction with the Generic Holdout, except potentially longer.

However, implementing the above thought experiment in practice is
infeasible:
generating the set of hypotheses up front requires the scientist to simulate their behavior in the (hypothetical) case that every hypothesis they propose to the holdout is false.
This is infeasible in settings where generating hypotheses is very expensive (in CPU-hours, or scientist-hours),
or settings where the scientist cannot properly simulate themselves.
Moreover, it is unclear how many hypotheses should be generated in such a way, before moving to the validation stage.
For example, consider large-scale physics experiments, where a scientific process is to first collect large amounts of data, and then explore the data to find interesting structures, and propose physical theories.
Here, the process of investigating the data and coming up with theories of physics is extremely expensive.
In this case, our proposal allows scientists to only invest effort in this process while they still have not derived a true hypothesis.
\end{remark}

{\bf Applications.}
The primary application of the Generic Holdout, as discussed above, is
to allow for adaptivity while preventing false discoveries --- in particular, as an alternative to pre-registration.
Here, our proposed method does not require the scientist to specify hypotheses
ahead of time, before analyzing the data.
Instead, the scientist only needs to specify how to determine if a hypothesis is significant or not.
Using the Generic Holdout, the scientist can use any data analysis method (valid, or not) on their exploration set, and as long as they check with the holdout mechanism before publishing, they will not publish a false hypothesis except with small probability.
Moreover, it is sample-efficient: they can ask up to exponentially-many false hypotheses to the holdout set, before a true hypothesis is confirmed.

Journals could even require that researchers submit their holdout set to the journal, without looking at it, and then journals implement the Generic Holdout mechanism themselves. That is, researchers (potentially interactively) submit hypotheses (with associated hypothesis tests) to the journals,
which respond with a single bit.

The Generic Holdout also naturally applies to any data analysis procedure that involves several steps, each of which needs to be validated. For example, suppose the data analyst would like to first check ``is the data well-clustered into 10 clusters?'', and then based on this, search for a good kernel embedding of the data, et. cetera.
The Generic Holdout allows for such procedures to be statistically sound, without requiring any understanding of the exact statistical properties of the analyst's queries.

Another application of the Generic Holdout is in fields where the existing scientific process appears to be working, but journals would like to have statistical soundness guarantees.
This is especially relevant when a large data set is collected once and made public, and many research groups subsequently investigate and publish findings about the data set (for example, as in Genome-Wide Association Studies~\cite{gwas-catalog, gwas}).
Here, the proposal is: journals request some holdout data from the group initially collecting the data set, and put it in a vault (without publishing it, or looking at it).
For every submitted study using the common data set, journals first do their usual review process. Then, when a paper passes their usual review, they do a final validation check on the holdout data.
In this setting, the Generic Holdout guarantees that the journal can validate exponentially-many true hypotheses, and the holdout is only extinguished once it catches several false hypotheses.
(Note, this is a complementary setting to the first application.)

{\bf Organization.}
In Section~\ref{sec:related}, we discuss lines of prior work on adaptivity in the sciences, related to our proposal.
We formally describe our scientific setup and goals in Section~\ref{sec:framework}, phrased in the language of \emph{statistical hypothesis testing}.
We provide an extended formal example of problems that arise due to adaptivity in data analysis in Section~\ref{sec:problems}.
In Section~\ref{sec:generic-holdout} we describe our proposed method, the Generic Holdout,
and formally state its statistical guarantees.
We include an example instantiation of our generic framework in Section~\ref{sec:example}, illustrating a common setting where exponentially-many hypotheses can be tested until several true ones are discovered.

\subsection{Related Works}
\label{sec:related}
There are many related works surrounding the problem of adaptivity in the sciences; we discuss and compare the most relevant ones below.

\paragraph{Reusable Holdout.}
The proposal most similar to ours is the Reusable Holdout \cite{reusable, reuse2},
which developed out of ideas from Differential Privacy~\cite{DP} and Adaptive Data Analysis~\cite{validity}.
The Reusable Holdout addresses a similar scientific problem --- preventing false discoveries in data analysis
--- and proposes a very similar methodology.
However, there are several key differences that allow us to improve on the Reusable Holdout in our setting.

The Reusable Holdout is a mechanism for interacting with the holdout set in (informally) the following way.
When the scientist proposes a hypothesis, the mechanism first checks if the hypothesis ``looks similar'' on the holdout and exploration sets
(i.e., if they have similar $p$-values, a measure of how well the hypothesis fit the data).
If they are indeed similar, the mechanisms essentially releases the $p$-value of the hypothesis on the exploration set, not involving the holdout data.
If they are very different, the mechanism releases a noisy $p$-value on the holdout set.
This mechanism leaks information about the holdout set (a noisy $p$-value) whenever the scientist proposes ``bad'' hypotheses,
which are overfit to the exploration data.
As a result, the Reusable Holdout can handle only quadratically-many ``bad'' hypotheses in the size of the holdout set.
In our setting, where the scientist may use an arbitrary exploration procedure to generate hypotheses, many of the proposed hypotheses may in fact be ``bad'', and the Reusable Holdout would quickly become unusable.
In contrast, our method allows for up to an exponential number of ``bad'' hypotheses,
as long as the scientist stops after discovering a few true ones.

We stress that the technical details of the Generic Holdout are not novel (e.g., the \textsc{SparseValidate} mechanism of \cite{reuse2} is essentially the complement of our mechanism), but we believe our formalization of the scientific problem is meaningful,
and our proposal cleanly solves this problem.

As an aside, note that the Reusable Holdout also releases estimated $p$-values of hypotheses, while the Generic Holdout
only releases binary responses. However, to solve the scientific problem of
preventing false discoveries, releasing binary responses is sufficient.
Moreover, if we would like $p$-values for the validated hypotheses, we can estimate these by simply keeping another small holdout set.
Finally, the Generic Holdout allows for testing a much more general class of
hypotheses than the Reusable Holdout,
and does not require any specialized analysis per hypothesis class.

\paragraph{Adaptive Data Analysis.}
The recently-developed theory of Adaptive Data Analysis \cite{reuse2, validity, stability, HU}
also addresses issues of how to do valid, adaptive science.
At a high level, the goal of Adaptive Data Analysis is much more ambitious than our goal:
there, the goal (informally) is to address the question of how to form good, generalizable hypotheses based on data.
In contrast, our goal is merely to prevent a scientist from being convinced of
false hypotheses; we do not give a procedure for deriving true hypotheses in the
first place.

More formally, in Adaptive Data Analysis we have some underlying distribution $D$ on the universe $U$,
and the scientist would like to (approximately) know the result of statistical queries
$\E_{X \sim D}[\phi_i(x)]$ for some adaptively-chosen sequence of queries $\phi_i : U \to [0, 1]$.
The mechanism has access only to samples $X_1, \dots X_n \sim D$,
and must answer the queries with an estimate $\hat{\mu_i} \approx \E_{X \sim D}[\phi_i(X)]$.
Thus, a scientist only interacting with the data through such a mechanism will always receive answers close to the
true answers on the distribution, and thus will not generate hypotheses which are overfit to the data.

The tools developed in this area give computationally-efficient mechanisms for providing such estimates.
However, due to the strong guarantees provided, these mechanisms can only correctly answer
quadratically-many queries in the size of the data set
\cite{stability}.
Moreoever, it is computationally intractable to answer more than polynomially-many queries in this setting \cite{HU}.

We note that the Generic Holdout is well-suited to be used in conjunction with the methods of Adaptive Data Analysis.
That is, we imagine a scientist using the tools of Adaptive Data Analysis (amongst possibly other methods)
to generate hypotheses using the exploration set,
and using the Generic Holdout to confirm them before publication.

\paragraph{Inference after Model Selection}
This recent line of work ~\cite{model1, model2} focuses on a specific kind of analysis procedure,
which proceeds in two stages (model selection, and inference given the model).
For example, if the analyst first selects influential variables via $L_1$-regularized regression
(``model selection'') and then forms hypotheses based on these variables (``inference'').
Adaptivity arises as a problem here for the same reason,
since model-selection and inference are both performed on the same data set, and thus the hypotheses are dependent on the data they are tested against.
This is essentially ``2 rounds of adaptivity'' in our setting.
For specific kinds of data distributions and model-selection procedures, these works are able to
precisely analyze how the hypotheses depend on the data,
and thus give bounds on the performance of the overall procedure.

Our proposal is more general, in that it allows for multiple stages of adaptivity, each stage of which could be arbitrary.
For example, our proposal would allow for model-selection in multiple stages, each of which needs to be validated with respect to the population distribution (e.g., ``find a good embedding, then select influential variables, then cluster according to these...'').
Moreover, our proposal makes no assumptions on the data distribution or hypothesis class,
and can handle cases that may be hard to fully understand in the ``inference after model selection'' framework.
Of course, for specific cases that can be understood, this framework could lead to tighter results than our generic framework.

\paragraph{Adaptive FDR Control.}
There is a related line of work that is interested in controlling the False-Discovery-Rate (FDR)
of hypothesis testing.
The setting here is as follows.
We have a fixed large set of hypotheses, and we want to simultaneously test all of them
on the same data set, while bounding the False-Discovery-Rate:
the fraction of false hypotheses among all hypotheses that passed validation
(i.e., the ``false discoveries'' among all discoveries).
The scientific motivation here is often that we want to prune our hypotheses to a small set of ``interesting'' ones, on which we will then conduct further independent testing.
For example, if we are interested in finding genes which cause a disease,
we may first test all the hypotheses ``gene X is correlated with the disease'' for every value of gene X,
using a testing procedure with bounded FDR.
Then, among the returned hypotheses, we can do further experiments to determine their effect (say, looking physically at the mechanism of the gene expression).
Here, controlling the FDR is important, since we do not want to invest too many resources into experiments which are likely to be null.

There are various proposed methods for controlling FDR in different settings~\cite{BH,knockoffs}, and in particular, recently
there have been proposals to control FDR by adaptively deciding the order in which to test hypotheses, based on the results of past hypothesis tests~\cite{adaPT, fdr1, fdr2}. This could potentially yield more \emph{powerful} tests, i.e. tests that are more likely to discover true hypotheses.

These works on [adaptive] FDR control operate in a different setting from our work,
first because their notion of ``adaptive'' is different (the large set of hypotheses is usually assumed to be fixed beforehand),
and second because they are interested in a different notion of error
(the controlling the false-discovery rate, instead of preventing false discoveries overall).
In the statistical terminology, our proposal controls the ``family-wise error rate (FWER)'' instead of the FDR.

\section{The Scientific Framework}
\label{sec:framework}
In this section, we define our scientific framework, in the language of \emph{hypothesis testing}.

There exists some universe $U$ and underlying true distribution $D$ on $U$, specified by Nature.
(For example, $U$ could be genomic sequences, and $D$ the true distribution of human genome.)

The scientist can form \emph{hypotheses} $H \in \H$ about the true distribution (eg, ``gene X is correlated to disease Y'').
Each hypothesis corresponds to a partition of set of all distributions into a Null class $\cD_{H}^{Null}$ and an Alternative class $\cD_H^{Alt}$.
(The Null class defines distributions where the hypothesis is false, and the Alternative where the hypothesis is true).

For each hypothesis $H$ in the hypothesis class $\H$, we have a \emph{hypothesis test} $\Test^{(n)}_H: U^n \to \{0, 1\}$ which takes $n$ independent samples from a distribution and is supposed to accept under distributions in $\cD_H^{Alt}$ and reject under those in $\cD_H^{Null}$.
The false-positive probability of each test is known as its \emph{$p$-value}, and is given by
$$p := \sup_{D \in \cD_H^{Null}} \left\{ \Pr_{X_1, \dots X_n \sim D}[\Test_H(X_1, \dots X_n) \text{ accepts}] \right\}$$

In classical (non-adaptive) science, the scientific process is:
we first fix some hypothesis $H$, and then collect independent samples
$X_1, X_2, \dots X_n \sim D$ from the true distribution $D$,
and run the hypothesis test $\Test^{(n)}_H(X_1, \dots X_n)$.

For a single fixed hypothesis, we are usually interested in controlling the
false-positive probability of the hypothesis test.
This gives evidence for believing in hypotheses which pass the hypothesis test, in the following sense: Suppose a hypothesis test for hypothesis $H$ has false-positive probability $p \ll 1$. Then, if the hypothesis were false, our experimental procedure would have invalidated it with large probability $(1-p)$.

The setting where we have a fixed set of hypotheses, and want to test them all simultaneously,
is known as \emph{multiple hypothesis testing}.
In this setting, we could want to control different notions of error --- for example,
controlling the overall probability of confirming a false hypothesis,
or controlling the fraction of false hypotheses among confirmed hypotheses.
Throughout this work, we will consider controlling the overall probability of confirming a false hypotheses (and further, our hypotheses will be generated \emph{adaptively}).

In particular, we consider the general adaptive scientific process as follows.
The scientist first collects a data set $X_1, X_2, \dots X_n \sim D$ of $n$ independent samples from $D$.
Then, the scientist is interested in exploring the data set to find true hypotheses, and will eventually propose a hypothesis (or small set of hypothesis) that s/he believes to be true.
\emph{We would like to guarantee that the finally proposed hypotheses are in fact true} --- that is, we want to bound the false-positive probability of the proposed hypotheses.

The Generic Holdout is a general, sample-efficient method to achieve this.

\subsection{The Problem with Adaptivity}
\label{sec:problems}
In this section, we give an extended formal example that illustrates the problem of adaptivity in data analysis
(a version of what is known a ``Freedman's paradox'' \cite{paradox}.)

Naively, if we collect a data set, form a hypothesis based on it, and then test the hypothesis on the same data set, we lose all guarantees of correctness.
This is essentially because if we are allowed to adapt to our data set (and choose among many hypotheses), we can easily ``overfit'' to our data set, and find some hypothesis that is true about the data but not true in Nature.
As an informal example, say we collect data on a set of 20 random people. Let their set of names be $S$.
Then we form the hypothesis ``At least 99\% of people have names in S.''
Clearly this hypothesis is well-supported by the data, but entirely false. Moreover, this hypothesis would be correctly rejected if it were formed \emph{a priori}, and tested on an independent set of people.

The above problem still exists if we do not look at the data set directly, but we are allowed to \emph{adaptively} choose hypotheses to test.
That is, as a scientist we are not committed to a set of hypotheses beforehand, but rather we
are interested in exploring the data set to find interesting structures.
So we will first test some hypothesis $H_0$ against the data set, and then seeing the results of this test (say, its $p$-value), we pick another hypothesis $H_1$ to test, and so on.
In the example below, we will see that this can easily lead to a scientist being convinced of a false hypothesis $H_k$,
which appears to be true on the data set (\textit{i.e.} passes validation with low $p$-value).
Roughly what happens is, the scientist will test a series of ``weak'' hypotheses, and seeing the results of these hypotheses tests, will combine them into a
single ``strong'' hypothesis which is over-fit to the data set.

{\bf Formal Example.}
Let us consider the universe $U := \R^{d+1}$, and distributions over $(x_1, \dots x_d, y) \in \R^{d+1}$

We will form a sequence of hypothesis $H_i$.
Each hypothesis is of the form: $y$ is positively correlated with $\langle w, x\rangle$ for $||w||_2 = 1$.
That is, each hypothesis $H^{(w)}$ is specified by $w$, and the Alternative class for $H^{(w)}$ corresponds to distributions on $(\vec{x}, y)$ for which $\E[y \cdot \langle w, x \rangle] > 0$.
(The Null class for $H^{(w)}$ is the complement of the Alternative class).

Note that the distribution where $(\vec{x}, y)$ are i.i.d. Gaussians $\mathcal{N}(0, 1)$ belongs to the Null class for all hypotheses. Call this distribution the ``Global Null.''

For a single, \emph{a priori} fixed hypothesis $H^{(w)}$, it is sufficient to take $n = O(\log(1/p))$ independent samples from the distribution in order to  test this hypothesis with false-positive probability $p$.
That is, the hypothesis test for $H^{(w)}$ takes $n$ samples $\{(x^{(i)}, y^{(i)})\}$, and tests if the empirical correlation $\left(\frac{1}{n}\sum_i y^{(i)} \langle w, x^{(i)} \rangle\right) > 1$.
This test has $p$-value $p$, meaning that under any distribution from the Null class of $H^{(w)}$, this test rejects except with probability $p$ (said another way, $p$ is the ``false-positive'' probability).

Similarly, for any \emph{a priori fixed} set of $k$ hypotheses, it is sufficient to take $n = O(\log(k/p))$  samples. In statistical parlance, this is equivalent to the ``Bonferroni Procedure'', i.e. the Union Bound, which says that to test $k$ fixed hypotheses simultaneously with error level $p$, one should test each individual hypothesis using at level $(p/k)$.

Now, suppose we are in the Global Null distribution, and consider the following scientist who is trying to find a positive hypothesis in the class defined above.
We will make only $k= d+1$ queries total, so we decide to take $n = O(\log(d/p))$ samples from the distribution (this is \emph{incorrect} as we will see, since it assumes our queries were fixed in advance).
For the first $d$ queries, the scientist tests the hypotheses $H^{(w_0)}, H^{(w_1)}, \dots H^{(w_d)}$ for $w_i = \vec e_i$ the $i$-th standard basis vector.
Knowing the $p$-values from these tests, the scientist knows the empirical correlations $c_i := \tilde{\E}[y x_i] := \frac{1}{n} \sum_j  y^{(j)} x^{(j)}_i$ between each of the coordinates $x_i$ and $y$ on the samples.
Each of these empirical correlations will have magnitude $\left| \tilde{\E}[y
x_i] \right| \gtrapprox 1/\sqrt{n}$ in expectation.
Now for the final query, the scientist checks the hypothesis $w^* := \frac{1}{\sqrt{d}}\text{sign}(c)$.
This has empirical correlation $\tilde{E}[y \langle w^*, x \rangle] > \frac{\sqrt{d}}{\sqrt{n}}$ by construction, since we sum all the coordinate-wise correlations. Note that with our choice of $n$, we have
$\frac{\sqrt{d}}{\sqrt{n}} \gg 1$, meaning this hypothesis test passes, even though we were in the Null distribution.

\textbf{Conclusions.} The above shows that methods to do hypothesis-testing with a fixed set of hypothesis (eg, controlling the $p$-values using the ``Bonferroni Procedure''/union bound) can fail catastrophically when these hypotheses are chosen adaptively, knowing the results of previous hypothesis tests.
In particular, a method for \emph{a priori} testing may require exponentially more samples to be correct for adaptive testing.
Note that this counterexample continues to hold if hypotheses are tested using cross-validation (ie, each hypothesis is tested on a different random subset of the data set.)

\begin{remark}
Looking closer at the above example, what happened is that the scientist first tested many ``weak'' hypotheses, which failed validation, but then combined the results of these weak hypotheses into a ``strong'' hypothesis, which passed validation. The Generic Holdout prevents such failures, by not releasing any additional information about weak hypotheses which do not validate.

The first (naive) example discussed in this section is a trivial manifestation of the problem with adaptivity, where the scientist is ridiculously malevolent. The second example, however, is much more enlightening and could serve as an abstraction for a mistake done by an honest, yet not careful enough scientist!
\end{remark}

\section{Proposed Method: The Generic Holdout}
\label{sec:generic-holdout}

We propose the following scientific methodology (the ``Generic Holdout'').
\begin{enumerate}
		\item Take $n$ independent samples, and partition them into a \emph{exploration set} and a \emph{holdout set}.
	\item Set aside the holdout set and never look at it directly.
	\item Use the exploration set freely, in any way, to adaptively explore and propose hypotheses.
	\item When you have a plausible hypothesis $H$ in hand,
	prepare a hypothesis test for $H$, with desired $p$-value, and apply this test on the holdout set, observing only the outcome of the test (whether it rejects the null hypothesis or not).\\
	\emph{It is crucial that the binary outcome of the test is the only information observed from the holdout set. One must not observe more information, for example the actual $p$-value of the test on the holdout set.}
	\item You are free to adaptively repeat steps 3 and 4 to discover small number of true hypotheses.
\end{enumerate}

\subsection{Statistical Guarantees of the Generic Holdout}

Here we set up some notation regarding the methodology proposed above that will be useful in further discussion. We consider some universe $U$, and collect a data set $U_1, \ldots U_n \in U$, assumed to be a sequence of independent samples from some underlying \emph{population} --- probability distribution $\cD$ over $U$. We partition it into $U_1, \ldots U_h$ --- the \emph{holdout set}, and $U_{h+1}, \ldots U_n$ --- the \emph{exploration set}. The scientist uses exploration set to propose hypotheses $H_1, H_2, \ldots, H_s$ together with tests $\Test_{H_i}$ for each of them --- each of hypotheses $H_i$ can depend arbitrarily on the exploration set, and on results of all the previous tests.

When a scientist commits to use this mechanism until the number of validated hypotheses exceeds specific threshold $k$, or number of hypotheses tested altogether exceeds some specific threshold $s$, we wish to give strong statistical guarantee on the false positive rate for validated hypotheses. We focus on the scenario where $k \ll s$, i.e. we wish to discover only several true hypotheses, and we show that in this situation, the necessary size of the holdout set to achieve a fixed false positive probability scales gracefully with the total number of trials $s$.

The choice of the size of \emph{exploration set} is not relevant to this discussion; clearly larger exploration set makes it easier for the scientist to produce valid hypotheses in the first place, but the acquisition and maintenance of larger data set is often related with additional costs.

We will now formally define the adaptive hypothesis selection mechanism.
\begin{definition}[Adaptive hypothesis selection]
\label{def:hyp-selection}
		We define the \emph{$k$-bounded adaptive hypothesis} selection to be a sequence of (randomized) functions $\Alg_1, \ldots \Alg_s$ such that $\Alg_i : U^{n - h} \times \{0, 1\}^{s-1} \to \mathcal{H} \cup \{ \bot \}$. We think of $\Alg_i$ as a randomized scheme specifying how to pick $H_i$, based on the exploration set, and results of all previous hypotheses tests. We assume that after finding $k$ valid hypotheses, the researcher stops exploration, i.e. $\Alg_i(\vec{U}, x_1, \ldots x_{i-1}) = \bot$ whenever there are $k$ ones among $x_1, \ldots x_{i-1}$.
\end{definition}

Our main theorem quantifies the false-positive guarantees of the generic holdout test.
\begin{theorem}
		\label{thm:main}
		Consider a sequence of hypotheses $H_1, H_2, H_3, \ldots H_{s} \in \mathcal{H} \cup \{\bot\}$ generated as in Definition~\ref{def:hyp-selection},
		that is, the scientist adaptively generates up to $s$ hypotheses, and stops once $k$ hypotheses are confirmed.
		If the $p$-value of each test $H_i$ is bounded by $p$, then probability of false discovery in this workflow is bounded by $s^k p$. More formally,
		\begin{align*}
				\forall D \in \mathcal{D},\
			    &\Pr[ ~\textnormal{Scientist accepts a false hypothesis}~ ]\\
			    &=
				\Pr_{U_1, \ldots U_n \sim D}\left(\exists i \leq s,\  D \in D_{H_i}^{Null} \land \Test_{H_{i}}(U_1, \ldots U_h) = 1\right)
			    \leq	s^k p.
		\end{align*}
\end{theorem}

The proof of this theorem is elementary, before we proceed with it let us state explicitly important interpretation of its statement.

\textbf{Discussion.} In order to achieve some target statistical significance, say $p_0 = 0.05$, over the whole process described above, we want to use holdout set such that the guaranteed false-positive probability $p$ for each specific test $H_i$ is of the order of $p_0 \left/ s^{k}\right.$. Often for standard statistical tests the required samples size scales like $\Oh(\log \frac{1}{p})$ with the desired $p$-value, and as such it is enough to use the holdout set of size $\Oh(k \log s)$.

To put it differently, once we have fixed holdout set of size $h$, desired $p$-value $p_0$ and bound $k$ on the number of discovered ``true'' hypotheses (after which we stop using collected holdout set for verification), we can issue $s = 2^{\Omega(h/k)}$ queries in the workflow described above, and still have confidence $p_0$ on the validity of all discovered hypotheses.

\textbf{Remark.}
For $k=1$, this bound exactly matches the ``Bonferroni Procedure'' (ie, the union bound) for testing a fixed set of $s$ non-adaptive hypotheses.

\textbf{Remark.} The statement of the theorem remains unaffected in the complementary setting, where we expect number of rejected hypotheses to be bounded by $k$. Here the scientist. Here, we expect scientist to use the mechanism until at most $k$ hypotheses are rejected, or at most $s$ queries are issued. In this scenario, we can again bound the probability of any false discovery by $s^{k} p$.

\textbf{Remark.} Note that simply providing a mechanism for validating hypotheses with small probability of
false discoveries is trivial: the mechanism can just respond that every hypothesis tested is false.
We would like mechanisms to also be \emph{useful}, in that they allow for true discoveries.
One possible formalization of usefulness guarantees of the Generic Holdout, for $k=1$, is as follows.
Intuitively, we want to say that a scientist who follows a strategy that eventually proposes a valid hypothesis, will discover this hypothesis while using the Generic Holdout. More formally, for a hypothesis $H \in \mathcal{H}$, distribution $D \in \cD$ and some associated test $\Test_H$ we define $p_{H, D} := \Pr_{X_1, \ldots X_h \sim D}(\Test_H(X_1, \ldots X_m) = 1)$. For $\Alg_1, \ldots \Alg_s$ as in Definition~\ref{def:hyp-selection}, and some distribution $D \in \cD$, we have
\begin{equation*}
    \Pr_{
    \substack{U_1, \dots U_h \sim D \\
    H_1, \dots H_s \gets \Alg(U_1, \dots U_h)}}
    (\exists i\leq s,\ \Test_{H_i}(U_1, \ldots U_h) = 1) \geq
    \E_{
    \substack{U_1, \dots U_h \sim D \\
    H_1, \dots H_s \gets \Alg(U_1, \dots U_h)}}
    \max_i p_{H_i, D}.
\end{equation*}

\begin{proof}[Proof of Theorem~\ref{thm:main}]
	Observe that, as $U_1, \ldots, U_h$ are assumed to be independent from $U_{h+1}, \ldots U_n$, and the internal randomness of the scientist. Let us, for now, assume that the selection of the $i$-th hypothesis depends only on the results of all previous tests $\Alg_i : \{0, 1\}^{i - 1} \to \mathcal{H} \cup \{\bot\}$ in a deterministic way.

	Note that for a fixed sequence $\Alg_1, \Alg_2, \ldots \Alg_s$ as above (i.e. we assume that $\Alg_{i}(x_1, \ldots x_{i-1}) = \bot$ if there are at least $k$ ones among $x_1, \ldots x_{i-1}$), there is at most $s \sum_{i < k} \binom{s}{k} \leq s^{k}$ hypotheses that will ever be tested by this algorithm --- this is a bound on the total range of all those functions. Consider the set $\tilde{\mathcal{H}} \subset \mathcal{H}$ given by union of all the ranges of $\Alg_i$. We know that $|\tilde{\mathcal{H}}| \leq s^{k}$, and moreover if we fix $D \in \cD$, we have
	\begin{multline*}
			\Pr_{U_1, \ldots U_h \sim D} (\exists i \leq s,\ H_i \in \cD_{H_i}^{Null} \land \Test_{H_i}(U_1, \ldots U_h) = 1) \\
		\begin{array}{l}
				\displaystyle
			 \leq \Pr_{U_1, \ldots, U_h \sim D} (\exists H \in \tilde{\mathcal{H}},\ \cD_{H}^{Null} \land \Test_{H_i} = 1) \\
			 \leq \displaystyle |\mathcal{\tilde{H}}| p\\
			 \leq \displaystyle s^{k} p.
		\end{array}
	\end{multline*}

	For general case, where $\Alg_i$ is a randomized function that depends also on the exploration set $U_{h+1}, \ldots U_{n}$, we can use the linearity of expectation --- conditioning on any deterministic realization of $\Alg_i$, and the value of exploration set $U_{h+1}, \ldots U_{n}$, the statement is true by the argument above, and therefore it is true, in expectation over those random variables.
\end{proof}
\subsection{Example: Gapped Empirical Losses}
\label{sec:example}

In many natural situations, the hypothesis test takes a special form: thresholding an \emph{empirical loss} evaluated on the sample at hand. Our general framework specializes to this case, and here we can give quantitative bounds on the number of samples $n$ required to bound false-positive rate.

Specifically, suppose that with each hypothesis $H$ we have some associated loss function $\ell_H: U \to [-1, 1]$, such that
$$\forall D \in \cD_H^{Null}: \E_{x \sim D}[\ell_H(x)] \leq 0,$$
moreover, suppose the hypothesis test is simply
\begin{equation}
    \Test_H(X_1, \dots X_h) = \mathbf{1}\{(\frac{1}{h}\sum_i \ell_H(X_i)) > 1/2 \}.
    \label{eq:threshold}
\end{equation}

In this case, we give quantitative bounds on the number of samples $n$ required for constant statistical confidence on validated hypotheses within the Generic Holdout framework.
\begin{theorem}
\label{thm:gap-loss}
If the scientist makes $s$ adaptive hypothesis test queries (generated as in Definition~\ref{def:hyp-selection}) on the holdout set, including at most $k$ that are confirmed to be valid, where each hypothesis test is of form \eqref{eq:threshold} then using holdout set of size $h = O(t\log(k/p_0))$ is sufficient to guarantee that the probability of confirming a false hypothesis is at most $p_0$.
\end{theorem}

One concrete realization of this class of hypothesis tests is following. Consider the class of multivariate normal distributions with covariance matrix bounded in spectral norm by $1$, and the problem of finding a linear predictor $h(x) := \langle w, x \rangle$ that is correlated with target feature $y$.  With each vector $w$ of unit norm, we can consider associated loss function $\ell_w(x, y) =  \mathrm{truncate}_{[-1, 1]}(\langle w, x \rangle)$, where $\mathrm{truncate}_{[-1, 1]}(x) = \min(\max(x, -1), 1)$. Hypotheses of this form can be generated by using linear regression on the exploration set, and then verified on the holdout set. Theorem~\ref{thm:gap-loss} states that we can validate exponentially many hypotheses (with respect to the size of given holdout set), as long as we stop upon discovering few true hypotheses of this form.

\section{Acknowledgements}
We would like to thank Madhu Sudan, Boaz Barak, Lucas Janson, Jonathan Shi, and Thibaut Horel for helpful discussions during the course of this work.

\bibliographystyle{plain}
\bibliography{refs}

\end{document}